\documentclass[12pt]{article}
\usepackage{amsmath}
\usepackage{graphicx,psfrag,epsf}
\usepackage{enumerate}
\usepackage{psfrag,epsf}
\usepackage{url} % not crucial - just used below for the URL
\usepackage{amsmath,amssymb,amsfonts,amsthm,mathtools,mathrsfs}
\usepackage{booktabs}

\usepackage{adjustbox}
\usepackage{tikz}
\usepackage{pgfplots}
\usepgfplotslibrary{fillbetween}
\pgfplotsset{compat=1.16}

\usepackage{bm,bbm}
\usepackage{color}
\usepackage{subfigure}
\usepackage{algorithm,algpseudocode}
\usepackage[symbol]{footmisc}
\usepackage{scalefnt}
\usepackage{authblk} % author and affiliations
\usepackage{multirow,centernot}
\usepackage[colorlinks=true, citecolor=blue, urlcolor=blue]{hyperref}
\usepackage{longtable}
\usepackage{booktabs}

\usepackage{natbib}
\usepackage{dsfont}
\usepackage[title]{appendix}
\usepackage{newtxtext,newtxmath} % Times Roman fonts (text and math)
\usepackage{mhchem} % for \ce macro

\input cyracc.def

\usepackage[left=1in,top=1.1in,right=0.5in,bottom=1in]{geometry}

\newtheorem{theorem}{Theorem}[section]

\newtheorem{corollary}[theorem]{Corollary}

\newtheorem{definition}{Definition}[section]
\newtheorem{lemma}[theorem]{Lemma}
\newtheorem{proposition}[theorem]{Proposition}
\newtheorem{remark}[theorem]{Remark}
\makeatletter
\def\@seccntformat#1{\@ifundefined{#1@cntformat}%
	{\csname the#1\endcsname\quad}%      default
	{\csname #1@cntformat\endcsname}%    enable individual control
}
\makeatother

\newif\ifShowComments
\ShowCommentstrue
\def\strutdepth{\dp\strutbox}
\def\druk#1{\strut\vadjust{\kern-\strutdepth
        {\vtop to \strutdepth{%
                \baselineskip\strutdepth\vss
                        \llap{\hbox{#1}\quad}\null}}}}

\title{\bf
Unifying the Hoover and Gini indices: Analytical, bias, and computational aspects
}

\author{
\text{Roberto Vila}$^{1}$\thanks{Corresponding author: Roberto Vila, email: {rovig161@gmail.com}
}\,,
\,
\text{Helton Saulo}$^{1,2}$ 
\,and \,
\text{Felipe Quintino}$^{1}$
\\
{\small $^{1}$ Department of Statistics, University of Brasilia, Brasilia, Brazil}\\
{\small $^{2}$ Department of Economics, Federal University of Pelotas, Pelotas, Brazil}\\
}
\begin{document}
	\maketitle 	
	\begin{abstract}
We propose a new family of inequality indices that bridges the Hoover index and the Gini coefficient. The measure is defined as the normalized expected absolute value of a convex combination of deviations from the mean and pairwise differences, providing a continuous interpolation between these two classical indices. We establish key theoretical properties, including scale invariance, boundedness, continuity, and compliance with the Pigou--Dalton transfer principle. Analytical representations are derived, allowing explicit evaluation under gamma distributions and leading to closed-form expressions involving incomplete gamma functions. From a statistical perspective, we study the plug-in estimator, obtaining a general expression for its expectation and explicit formulas for its bias under gamma populations. Simulation results indicate good finite-sample performance, with decreasing bias and mean squared error as the sample size increases. An empirical application to GDP per capita data illustrates the practical usefulness of the proposed index as a flexible tool for inequality analysis.
	\end{abstract}
%	

%\clearpage
%{
%	\hypersetup{linkcolor=black}
%	\tableofcontents
%}

%-------------------------------------------------------------------
%
%
%Let $X\geqslant 0$ be a random variable with finite mean
%$
%\mu=\mathbb{E}(X)>0.
%$
%Let $X_1,X_2$ be independent copies of $X$.
%
%The Hoover (Robin Hood) index \citep{Hoover1936} and the Gini coefficient \citep{Gini1936} are respectively defined by
%\[
%H=\frac{1}{2\mu}\, \mathbb{E}|X-\mu|,
%\quad
%G=\frac{1}{2\mu}\, \mathbb{E}|X_1-X_2|.
%\]
%
%
%
%%This index generates a continuous family of inequality measures connecting the
%%Hoover and Gini indices.
%
%For $0<\lambda<1$, the index $I_\lambda$ measures inequality through the
%normalized expected magnitude of a convex combination of two deviations:
%the deviation from the mean and the pairwise difference between two
%independent observations. Hence, $I_\lambda$ captures both dispersion
%around the mean and interpersonal disparities, providing a continuous
%interpolation between the Hoover index ($\lambda=0$) and the Gini
%coefficient ($\lambda=1$).

\section{Introduction}

Measures of economic inequality play a central role in economics, statistics, and the social sciences. Quantifying disparities in income and wealth distributions is essential for understanding economic development, social welfare, and the effects of public policy. Among the most widely used measures of inequality are the Hoover index \citep{Hoover1936}, also known as the Robin Hood index, and the Gini coefficient \citep{Gini1936}. These classical indices have been extensively studied and applied in both theoretical and empirical contexts.

Let $X \geqslant 0$ be a random variable representing income with finite mean $\mu = \mathbb{E}[X] > 0$. Let $X_1$ and $X_2$ be independent copies of $X$. The Hoover index and the Gini coefficient are respectively defined as
\[
H = \frac{1}{2\mu}\, \mathbb{E}|X - \mu|, \quad
G = \frac{1}{2\mu}\, \mathbb{E}|X_1 - X_2|.
\]
The Hoover index measures inequality through the mean absolute deviation from the mean and admits a direct redistributive interpretation, representing the proportion of total income that would need to be reallocated to achieve equality. In contrast, the Gini coefficient is based on pairwise differences between individuals and is closely related to the Lorenz curve \citep{Lorenz1905}, providing a global measure of dispersion across the entire distribution.

Although both indices quantify inequality, they emphasize different aspects of dispersion. The Hoover index focuses on deviations from the mean, while the Gini coefficient captures interpersonal differences. As a consequence, these measures may lead to different conclusions in empirical applications, particularly in distributions with high skewness or extreme values \citep{cowell2007extreme}.

In this paper, we introduce a new family of inequality indices that continuously bridges these two classical measures. For $\lambda \in [0,1]$, we define
\begin{equation}\label{eq:index}
    I_\lambda = \frac{1}{2\mu} \, \mathbb{E}\left| (1-\lambda)(X_1 - \mu) + \lambda(X_1 - X_2) \right|.
\end{equation}
By construction, $I_0 = H$ and $I_1 = G$, so that the proposed family provides a continuous interpolation between the Hoover and Gini indices. For intermediate values of $\lambda$, the index combines dispersion around the mean with pairwise differences, yielding a flexible framework for inequality measurement.

The proposed index offers several advantages. First, it allows for a unified treatment of two classical measures within a single parametric family. Second, it provides a tool for sensitivity analysis with respect to the concept of inequality itself, enabling researchers to assess how conclusions depend on the relative importance assigned to mean deviations versus interpersonal differences. Third, the index admits tractable analytical representations that facilitate both theoretical analysis and statistical inference.

The contributions of this paper are threefold. First, we establish the main theoretical properties of the proposed index, including scale invariance, boundedness, continuity in $\lambda$, and compliance with the Pigou--Dalton transfer principle. Second, we derive analytical representations that allow explicit evaluation of the index for gamma distributions in terms of incomplete gamma functions. Third, we study the statistical properties of the natural plug-in estimator, deriving a general expression for its expectation based on Laplace transform techniques and exponential tilting, and obtaining explicit formulas for its bias under gamma populations.

The remainder of the paper is organized as follows.
In Section \ref{ssection:2} we introduce the proposed index and study its main theoretical properties.
In Section \ref{An auxiliary lemma} we present an auxiliary lemma based on exponential tilting and Laplace transform techniques that will be used in subsequent derivations.
In Section \ref{Bias} we derive general expressions for the expected value of the plug-in estimator of the index for non-negative populations.
In Section \ref{Bias of the estimator} we specialize these results to gamma distributions and obtain an explicit expression for the bias of the estimator.
In Section \ref{Monte Carlo simulation study} we report results from a Monte Carlo simulation study.  Section \ref{empirical application} presents an empirical application. Finally, Section \ref{Concluding remarks} concludes the paper.

\clearpage

\section{Properties of the index}\label{ssection:2}

In this section, for the new index $I_\lambda$ defined in \eqref{eq:index}, we present a general integral form and a closed-form expression for gamma families. 
Additionally, we list several properties, such as boundary cases, non-negativity, range, an upper bound index, scale invariance, and continuity in $\lambda$.

\begin{proposition}\label{prop-in}
	The index $I_\lambda$ of a non-negative, non-degenerate random variable $X$ with finite mean $\mathbb{E}[X]=\mu>0$ can be expressed as
\begin{align*}
	I_\lambda
	=
	1
	-
	{1\over \mu}
	\int_0^\infty
	\overline{F}_X(t^-)
	\overline{F}_{\lambda X_2+(1-\lambda)\mu}(t^-)	\,
	{\rm d}t
	=
	{1\over \mu}
	\int_0^\infty
	F_X(t^-)
	\overline{F}_{\lambda X_2+(1-\lambda)\mu}(t^-)	\,
	{\rm d}t,
\end{align*}
where $F_X$ denotes the distribution function of $X$ , 
$\overline{F}_X(x)=1-F_X(x)$, and  \(F_X(x^-)=\lim_{u\uparrow x}F_X(u)=\mathbb{P}(X<x)\).
\end{proposition}
\begin{proof}
Since
\[
(1-\lambda)(X_1-\mu)+\lambda(X_1-X_2)
=
X_1-\big[\lambda X_2+(1-\lambda)\mu\big],
\]
and using the identity
\begin{align}\label{id-fund}
|y_1-y_2|
=
\int_0^\infty
\left(
\mathbf{1}_{\{y_1\geqslant t\}}
+
\mathbf{1}_{\{y_2\geqslant t\}}
-
2\,\mathbf{1}_{\{y_1\geqslant t\}}\mathbf{1}_{\{y_2\geqslant t\}}
\right)
{\rm d}t,
\quad y_1,y_2\geqslant 0,
\end{align}
together with Tonelli's theorem and the independence of $X_1$ and $X_2$, we obtain
\begin{align*}
	I_\lambda
	=
	\frac{1}{2\mu}
	\int_0^\infty
	\left[
	\overline{F}_X(t^-)
	+
	\overline{F}_{\lambda X_2+(1-\lambda)\mu}(t^-)
	-
	2\,\overline{F}_X(t^-)\,
	\overline{F}_{\lambda X_2+(1-\lambda)\mu}(t^-)
	\right]
	{\rm d}t.
\end{align*}

Finally, using the well-known identity
\begin{align}\label{id-exp-posit}
\mathbb{E}[X]=\int_0^\infty \overline{F}_X(x^-)\,{\rm d}x,
\quad X\geqslant 0,
\end{align}
the result follows.
\end{proof}

\begin{proposition}\label{prop-Hoover-index}
	The index $I_\lambda$, with $\lambda\in(0,1]$, of a gamma random variable $X\sim\text{Gamma}(\alpha,\beta)$
	is given by
\[
I_\lambda
=
\frac{(1-\lambda)^\alpha\alpha^{\alpha-1} \exp\{-(1-\lambda)\alpha\} }{\Gamma(\alpha)}
+
\lambda\, \frac{\Gamma(\alpha,(1-\lambda)\alpha)}{\Gamma(\alpha)}
-
{1\over \alpha}
\int_{(1-\lambda)\alpha}^\infty
{\Gamma(\alpha, t)\over\Gamma(\alpha)} \,
{\Gamma\left(\alpha,{ t-(1-\lambda)\alpha\over\lambda}\right)\over\Gamma(\alpha)}	\,
{\rm d}t,
\]
where $\Gamma (s,x)=\int _{x}^{\infty }t^{s-1}\exp\{-t\}{\rm d}t$ is the upper incomplete gamma function and $\Gamma(s)=\lim_{x\to 0^+}\Gamma (s,x)$ is the (complete) gamma function.
\end{proposition}
\begin{proof}
%When $\lambda=0$, $I_0=H$ (see Section \ref{ssection:2}). For gamma populations, it is well known that \cite[Proposition 2.4]{Vila-Saulo2026}
%	\begin{align*}
%	H=\frac{\alpha^{\alpha-1}\exp\{-\alpha\}}{\Gamma(\alpha)}.
%	\end{align*}
%	
%We now consider the case $0<\lambda\leqslant 1$.
Since $X \sim \mathrm{Gamma}(\alpha,\beta)$, we have $\overline{F}_X(x)=\Gamma(\alpha,\beta x)/\Gamma(\alpha)$,
\begin{align*}
	\overline{F}_{\lambda X_2+(1-\lambda)\mu}(t)
	=
	\begin{cases}
	1, & t\leqslant (1-\lambda)\mu,
	\\
	\overline{F}_X\left({t-(1-\lambda)\mu\over\lambda}\right), & t> (1-\lambda)\mu,
	\end{cases}
\end{align*}
and
$\mu=\mathbb{E}[X]={\alpha}/{\beta}$. 
Applying Proposition \ref{prop-in}  and the scale property of the gamma distribution, we have
\begin{align*}
	I_\lambda
	=
	1
	-
	{\beta\over \alpha}
	\left[
	\int_0^{(1-\lambda)\alpha/\beta}
	{\Gamma(\alpha, \beta t)\over\Gamma(\alpha)}	\,
	{\rm d}t
	+
	\int_{(1-\lambda)\alpha/\beta}^\infty
	{\Gamma(\alpha, \beta t)\over\Gamma(\alpha)} \,
	{\Gamma\left(\alpha,{\beta t-(1-\lambda)\alpha\over\lambda}\right)\over\Gamma(\alpha)}	\,
	{\rm d}t
	\right].
\end{align*}

Finally, applying the change of variables $t \mapsto \beta t$ together with the identity 
\[
\int_0^{(1-\lambda)\alpha} 
\frac{\Gamma(\alpha,t)}{\Gamma(\alpha)}\,{\rm d}t
=
\alpha - \frac{(1-\lambda)^\alpha\alpha^\alpha \exp\{-(1-\lambda)\alpha\} }{\Gamma(\alpha)}
-
\alpha\lambda\, \frac{\Gamma(\alpha,(1-\lambda)\alpha)}{\Gamma(\alpha)},
\]
the proof follows.
\end{proof}

\begin{remark}
Note that $I_\lambda$ is scale invariant and therefore does not depend on the rate parameter $\beta$ (see Proposition \ref{prop-Hoover-index}).
\end{remark}

\begin{remark}\label{rem-H}
Since 
\[
\lim_{\lambda\to 0^+}
\mathbf{1}_{\{t>(1-\lambda)\alpha\}} \, 
\frac{\Gamma\left(\alpha,\frac{t-(1-\lambda)\alpha}{\lambda}\right)}{\Gamma(\alpha)}
=
\begin{cases}
	1, & t=\alpha,\\
	0, & t\neq\alpha,
\end{cases}
\]
it follows from Proposition \ref{prop-Hoover-index} that
\[
H
=
\lim_{\lambda\to 0^+}
I_\lambda
=
\frac{\alpha^{\alpha-1}\exp\{-\alpha\}}{\Gamma(\alpha)}.
\]
The above result was recently established as Proposition 2.4 in \cite{Vila-Saulo2026}.
\end{remark}

\begin{remark}\label{rem-G}
	Letting $\lambda\to 1$ in Proposition \ref{prop-Hoover-index}, we obtain
	\begin{align}\label{ide-last}
	G
	=
	\lim_{\lambda\to 1}
	I_\lambda
	=
	1
	-
	{1\over \alpha}
	\int_0^\infty
	{\Gamma^2(\alpha, t)\over\Gamma^2(\alpha)}	\,
	{\rm d}t
	=
	\frac{\Gamma\left(\alpha+\tfrac12\right)}
	{\sqrt{\pi}\alpha\Gamma(\alpha)}.
		\end{align}
The above result is well documented in the literature; see, for instance, \cite{Baydil2025}.
\end{remark}

\begin{remark}\label{remark:J_index}
%\textbf{Upper bound.}	
	Using the triangle inequality,
	$
	\left|
	(1-\lambda)(X-\mu)+\lambda(X_1-X_2)
	\right|
	\leqslant
	(1-\lambda)|X-\mu|+\lambda|X_1-X_2|.
	$
	Taking expectations yields 
	\begin{align}\label{ineq-H-G}
	I_\lambda\leqslant (1-\lambda)H+\lambda G\equiv J_\lambda.
	\end{align}
The $J_\lambda$ in the right-hand side of inequality \eqref{ineq-H-G} is much easier to manipulate mathematically than $I_\lambda$ defined in \eqref{eq:index}. A natural question, then, is why not use the right-hand side as the index, since it also bridges the Hoover and Gini indices.
In Section \ref{Monte Carlo simulation study}, we will show that the bias of the natural estimator of $I_\lambda$ has a smaller bias than the natural estimator of $J_\lambda$. 
\end{remark}

\begin{remark}
    	The index $I_\lambda$ satisfies the following properties:

\begin{enumerate}
	
	\item %\textbf{Boundary cases.}
	Note that $I_0=H$ and $I_1=G.$
	Thus, the proposed family smoothly connects the Hoover and Gini measures.
	
	\item %\textbf{Non-negativity.} 
%		\item %\textbf{Range.}
	Since $H\leqslant G$, the previous inequality \eqref{ineq-H-G} implies that the index remains bounded
	between zero and the Gini coefficient. That is,
	$
	0\leqslant I_\lambda\leqslant G\leqslant 1.
	$
%	
%    It is clear that
%	$I_\lambda\geqslant 0$. 
	Moreover,
	%\textbf{Perfect equality.}
	if $X=c$ almost surely for some constant $c$, then
	$
	I_\lambda=0.
	$
	
	\item %\textbf{Scale invariance.}
	For any constant $a>0$, let $Y=aX$. Then
	$
	I_\lambda(Y)=I_\lambda(X).
	$
	Hence, the index depends only on relative income differences. Furthermore, 
	%\textbf{Translation behavior.}
	if $Y=X+c$ with $c>0$, then
	$
	I_\lambda(Y)=[{\mu}/({\mu+c})] I_\lambda(X).
	$
	
	\item %\textbf{Continuity in $\lambda$.}
	The function $\lambda\mapsto I_\lambda$ is continuous on $[0,1]$.
	
%	The expression inside the expectation is a continuous function of $\lambda$,
%	and expectation preserves continuity.
	
	\item %\textbf{Symmetry.}	
	Since $X_1$ and $X_2$ are exchangeable, replacing them does not change the
	value of $I_\lambda$. In other words, the index depends only on the distribution of $X$.
	
	\item %\textbf{Covariance representation.} By using the identity
	Since
	$\mathbb{E}\vert Y\vert=2\textrm{Cov}(Y,\mathbf{1}_{\{Y>0\}})$ 
	for any centered variable $Y$, with $Y=(1-\lambda)(X_1-\mu)+\lambda(X_1-X_2)$, we obtaing
	\[
	I_\lambda
	=
	\frac{1}{\mu}\,
	\mathrm{Cov}\left(
	(1-\lambda)(X_1-\mu)+\lambda(X_1-X_2),
	\mathbf{1}_{\{(1-\lambda)(X_1-\mu)+\lambda(X_1-X_2)>0\}}
	\right).
	\]
	
	\item 
%	The index $I_\lambda$ satisfies the Pigou--Dalton principle of transfers.%
	%
%	\textit{Explanation.}
%	
%	Consider two individuals with incomes $x_i$ and $x_j$ such that
%	$x_i>x_j$. A progressive transfer consists of transferring a small
%	amount $\delta>0$ from the richer individual to the poorer one,
%	producing the new incomes
%	
%	\[
%	x_i' = x_i-\delta, \qquad x_j' = x_j+\delta,
%	\]
%	
%	with $0<\delta<({x_i-x_j})/{2}$ so that the income ranking is preserved.
	%
	If $X'$ denotes the income distribution after the transfer, then
	$
	I_\lambda(X') \leqslant I_\lambda(X)
	$ ({Pigou--Dalton transfer principle}).
%
%	Therefore, the index $I_\lambda$ satisfies the Pigou--Dalton principle of transfers.
%	
%	\textit{Proof.}
%	
%	The transfer reduces the dispersion of the income distribution while
%	keeping the mean unchanged. In particular,
%	$
%	\mathbb{E}(X')=\mathbb{E}(X)=\mu.
%	$
%	
%	Moreover, a progressive transfer decreases both absolute deviations from
%	the mean and pairwise income differences. Hence
%	
%	\[
%	\mathbb{E}|X'-\mu| \leqslant \mathbb{E}|X-\mu|
%	\quad
%	\text{and}
%	\quad
%	\mathbb{E}|X_1'-X_2'| \leqslant \mathbb{E}|X_1-X_2|.
%	\]
%	
%	Since the function
%	$
%	y \mapsto |y|
%	$
%	is convex, the quantity
%	
%	\[
%	\left|
%	(1-\lambda)(X-\mu)+\lambda(X_1-X_2)
%	\right|
%	\]
%	
%	is reduced under a mean-preserving contraction of the distribution.
%	
%	Taking expectations yields
%	
%	\[
%	I_\lambda(X')
%	=
%	\frac{1}{2\mu}\,
%	\mathbb{E}
%	\left|
%	(1-\lambda)(X'-\mu)+\lambda(X_1'-X_2')
%	\right|
%	\leqslant
%	I_\lambda(X).
%	\]
%	
%	Therefore, a progressive transfer from a richer individual to a poorer
%	individual decreases the value of the index, which establishes that
%	$I_\lambda$ satisfies the Pigou--Dalton transfer principle.
\end{enumerate}

\end{remark}

\section{An auxiliary lemma}\label{An auxiliary lemma}

The following  auxiliary lemma will be instrumental in deriving the results of Section \ref{Bias}.
\begin{lemma}\label{lemma-main}
If $W$, $Y$, and $Z$ are independent, non-degenerate, nonnegative random variables, then for 
$a,b,c \geqslant 0$ and $z>0$,
\begin{multline*}
	\mathbb{E}
	\left[
	\left|
	aW+bY-cZ
	\right|
	\exp\{-z(W+Y+Z)\}
	\right]
	\\[0.2cm]
	=
	\mathscr{L}_W(z)\mathscr{L}_Y(z)\mathscr{L}_Z(z)
	\left\{
	a\mathbb{E}[W_z]
	+
	b\mathbb{E}[Y_z]
	+
	c\mathbb{E}[Z_z]
	\right\}
	\pmb{G}(aW_z+bY_z,cZ_z),
\end{multline*}
where $\mathscr{L}_W(z)=\mathbb{E}[\exp\{-zW\}]$ is the Laplace transform of $W$, and for independent random variables $Y_1$ and $Y_2$,
\begin{align}\label{Normalized Mean Absolute Difference}
\pmb{G}(Y_1,Y_2)
=
\frac{\mathbb{E}\lvert Y_1-Y_2\rvert}
{\mathbb{E}[Y_1]+\mathbb{E}[Y_2]}
\end{align}
denotes the \textbf{normalized mean absolute difference}.
In the particular case where $Y_1$ and $Y_2$ are independent copies of a random variable $Y$, 
the quantity $\pmb{G}(Y_1,Y_2)$ reduces to the Gini coefficient $G=G(Y)$ of $Y$.

Furthermore, for each $z>0$, $W_z$ denotes the exponentially tilted (or Esscher-transformed) 
version of $W$ \citep{Butler2007}, whose cumulative distribution function (CDF) is given by
\begin{align}\label{def-F-X}
	F_{W_z}(t)
	=
	\frac{
		\mathbb{E}\left[
		\mathds{1}_{\{0\leqslant W\leqslant t\}}
		\exp\{-zW\}
		\right]
	}
	{\mathscr{L}_W(z)},
	\quad t\geqslant 0.
\end{align}

In the above, we assume that the Laplace transforms and expectations involved exist and are finite.
\end{lemma}
\begin{proof}
Since $W$, $Y$, and $Z$ are independent nonnegative random variables, for any 
Borel-measurable function $g:(0,\infty)^3\to\mathbb{R}$ it follows from the 
definition of the exponentially tilted distribution that
\[
\mathbb{E}[g(W_z,Y_z,Z_z)]
=
\frac{\mathbb{E}\left[g(W,Y,Z)\exp\{-z(W+Y+Z)\}\right]}
{\mathscr{L}_W(z)\mathscr{L}_Y(z)\mathscr{L}_Z(z)} .
\]
Hence,
\begin{align*}\label{iid-1}
	\mathbb{E}
	\left[
	\left|
	aW+bY-cZ
	\right|
	\exp\{-z(W+Y+Z)\}
	\right]
	=
	\mathscr{L}_W(z)\mathscr{L}_Y(z)\mathscr{L}_Z(z)\,
	\mathbb{E}
	\left|
	aW_z+bY_z-cZ_z
	\right|.
\end{align*}
Since $aW_z+bY_z$ and $cZ_z$ are independent, the result follows directly from the definition of the normalized mean absolute difference between $aW_z+bY_z$ and $cZ_z$.
\end{proof}
%\begin{lemma}\label{lemma-main-0}
%Under the assumptions of Lemma \ref{lemma-main}, it holds that
%	\begin{multline*}
%		\mathbb{E}
%	\left[
%	\left|
%	aW+bY-cZ
%	\right|
%	\exp\{-z(W+Y+Z)\}
%	\right]
%	\\[0,2cm]
%	=
%	\mathscr{L}_W(z)\mathscr{L}_Y(z)\mathscr{L}_Z(z)
%\left\{	
%a\mathbb{E}[W_z]
%+
%b\mathbb{E}[Y_z]
%+
%c\mathbb{E}[Z_z]
%-
%2
%\int_0^{\infty}
%\overline{F}_{aW_z+bY_z}(t^-)
%\overline{F}_{cZ_z}\left(t^-\right)
%{\rm d}t
%\right\}.
%	\end{multline*}
%\end{lemma}
%

\section{Expectation of the estimator $\widehat{I}_\lambda$}\label{Bias}

In this section we derive simple expressions for the expected value of the estimator $\widehat{I}_\lambda$ for general non-negative-support populations (Theorem \ref{main-theorem}). This result allows us to derive the estimator's bias in the gamma case (see Corollary \ref{main-corollary} in Section \ref{Bias of the estimator}).

Let $X_1,\dots,X_n$ be a sample from the distribution of $X$ with sample mean
\[
\overline{X}=\frac{1}{n}\sum_{i=1}^n X_i .
\]

The plug-in estimator of $I_\lambda$ is obtained by replacing $\mu$ with
$\overline{X}$ and the expectation with the corresponding sample average. Thus, for $\lambda\in[0,1]$,
\begin{align}\label{def-I-lambda}
\widehat I_\lambda
=
\frac{1}{2n(n-1) \overline{X}}
\sum_{1\leqslant i\neq j\leqslant n}
|
(1-\lambda)(X_i-\overline{X})+\lambda(X_i-X_j)
|
\,
\mathds{1}_{\{\sum_{i=1}^{n}X_i>0\}}.
\end{align}

For $\lambda=0$, this estimator reduces to the Hoover estimator \citep{Vila-Saulo2026}
\begin{equation}\label{eq:Hoover_estimator}
    \widehat I_0
\equiv
\widehat{H}
=
\frac{1}{2n \overline{X}}
\sum_{i=1}^n |X_i-\overline{X}|
\,  \mathds{1}_{\{\sum_{i=1}^{n}X_i>0\}},
\end{equation}
while for $\lambda=1$ it becomes the upward-adjusted Gini estimator \citep{Deltas2003}
\begin{equation}\label{eq:Gini_estimator}
    \widehat I_1
\equiv
\widehat{G}
=
\frac{1}{n(n-1) \overline{X}}
\sum_{1\leqslant i<j\leqslant n}|X_i-X_j| 
\,  \mathds{1}_{\{\sum_{i=1}^{n}X_i>0\}}.
\end{equation}

\begin{remark}\label{remark:J_index-1}
		Using the triangle inequality,
	$
	|
	(1-\lambda)(X_i-\overline{X})+\lambda(X_i-X_j)
	|
	\leqslant
	(1-\lambda)|X_i-\overline{X}|+\lambda|X_i-X_j|.
	$
Summing over all \( i, j \) such that \( 1 \leqslant i \neq j \leqslant n \), and then multiplying by \( \big[2n(n-1)\,\overline{X}\big]^{-1} \), yields
	\begin{align}\label{ineq-H-G-1}
	\widehat I_\lambda\leqslant (1-\lambda)\widehat{H}+\lambda \widehat{G} \equiv 	\widehat J_\lambda.
	\end{align}
\end{remark}

\begin{theorem}\label{main-theorem}
	Let $X_1, X_2, \ldots$ be independent copies of a non-negative and non-degenerate random variable $X$ with finite mean $\mathbb{E}[X]=\mu>0$. The following holds:
	\begin{align*}
	\mathbb{E}\left[\widehat{I}_\lambda\right]
=
n\left(1+{\lambda-1\over n}\right)
	\int_{0}^{\infty}
\mathscr{L}_X^n(z)
\mathbb{E}[X_z]
\pmb{G}\left({1-\lambda\over n} \sum_{k=3}^{n}(X_k)_z+\left[\lambda+{1-\lambda\over n}\right](X_2)_z, \left[1-{1-\lambda\over n}\right](X_1)_z\right)
{\rm d}z,
	\end{align*}
where $\mathscr{L}_X(z)=\mathbb{E}\left[\exp\{-zX\}\right]$ denotes the Laplace transform of $X$, and $\pmb{G}(Y_1,Y_2)$ is the normalized mean absolute difference between $Y_1$ and $Y_2$ stated in \eqref{Normalized Mean Absolute Difference}. For each $z>0$, $F_{X_z}$ is the CDF of the exponentially tilted random variable $X_z$ defined in \eqref{def-F-X}.
Moreover, throughout the above derivations, it is implicitly assumed that all Lebesgue--Stieltjes and improper integrals involved are well defined.
\end{theorem}
\begin{proof}
Using the well-known identity
$
x \int_{0}^{\infty} \exp\{-zx\}\,{\rm d}z = 1, \, x>0,
$
with $x=\sum_{i=1}^{n} X_i$, we obtain
{\small
	\begin{align}\label{eq-1}
		&\mathbb{E}\left[
		\frac{1}{2(n-1)\sum_{i=1}^{n}X_i} 
		\sum_{1\leqslant i\neq j\leqslant n}
		|
		(1-\lambda)(X_i-\overline{X})+\lambda(X_i-X_j)
		|
		\right]
				\nonumber
		\\[0,2cm]
		&=
				\mathbb{E}\left[
		\frac{1}{2(n-1)\sum_{i=1}^{n}X_i}
		\sum_{1\leqslant i\neq j\leqslant n}
\left|
{(1-\lambda)\over {n}}
\sum_{\substack{k=1\\ k\neq i,j}}^n X_k
-
{\left[1-{(1-\lambda)\over n}\right]}
X_i
+
{\left[\lambda+{(1-\lambda)\over n}\right]}
X_j
\right|
		\right]
				\nonumber
\\[0,2cm]
&=
\frac{1}{2(n-1)} \!\!
\sum_{1\leqslant i\neq j\leqslant n} \!\!
			\mathbb{E}\left[
\int_{0}^{\infty}
\left|
{(1-\lambda)\over n} \!\!
\sum_{\substack{k=1\\ k\neq i,j}}^n \! X_k
\!-
\left[1-{(1-\lambda)\over n}\right] X_i
\!+
\left[\lambda+{(1-\lambda)\over n}\right]
X_j
\right|
\exp\left\{-z\left(\sum_{\substack{k=1\\ k\neq i,j}}^n X_k
\!+X_i\!+X_j\right)\right\}
{\rm d}z
		\right]
		\nonumber
\\[0,2cm]
&=
\frac{1}{2(n-1)} \!\!
\sum_{1\leqslant i\neq j\leqslant n} \!
\int_{0}^{\infty} \!\!
\mathbb{E}\!\left[
\left|
{(1-\lambda)\over n} \!\!
\sum_{\substack{k=1\\ k\neq i,j}}^n X_k
\!-
\left[1-{(1-\lambda)\over n}\right] X_i
\!+
\left[\lambda+{(1-\lambda)\over n}\right]
X_j
\right|
\exp\left\{-z\left(\sum_{\substack{k=1\\ k\neq i,j}}^n X_k
\!+X_i\!+X_j\right)\right\}
\right]
{\rm d}z,
	\end{align}
}\noindent
	where in the last equality Tonelli's theorem was used. 
	
Since $X_1,X_2,\ldots$ are independent and identically distributed (i.i.d.), it is clear that, for each $i=1,\ldots,n$, the variables $\sum_{\substack{k=1\\ k\neq i,j}}^n X_k$, $X_i$ and $X_j$ are independent, and
\begin{align*}
	\sum_{\substack{k=1\\ k\neq i,j}}^n X_k\stackrel{d}{=}\sum_{k=3}^n X_k,
\end{align*}
where $\stackrel{d}{=}$ denotes equality in distribution. Therefore, from \eqref{eq-1} we have
	\begin{multline}\label{pre-id-0}
			\mathbb{E}\left[
	\frac{1}{2(n-1)\sum_{i=1}^{n}X_i} 
	\sum_{1\leqslant i\neq j\leqslant n}
	|
	(1-\lambda)(X_i-\overline{X})+\lambda(X_i-X_j)
	|
	\right]
	\\[0,2cm]
	=
	\frac{n}{2} 
	\int_{0}^{\infty}
	\mathbb{E}\left[
	\left|
	{(1-\lambda)\over n}
	\sum_{k=3}^n X_k
	-
	\left[1-{(1-\lambda)\over n}\right] X_1
	+
	\left[\lambda+{(1-\lambda)\over n}\right] X_2
	\right|
	\exp\left\{-z\left(\sum_{k=3}^n X_k
	+X_1+X_2\right)\right\}
	\right]
	{\rm d}z.
	\end{multline}
	
	By using Lemma \ref{lemma-main} with 
	\begin{align}\label{notations}
	W=\sum_{k=3}^n X_k, \quad 
	Y=X_1, \quad 
	Z=X_2, \quad 
	a={(1-\lambda)\over n}, \quad 
	b=1-{(1-\lambda)\over n}, \quad 
	c=\lambda+{(1-\lambda)\over n},
	\end{align}
	we get
	\begin{align}\label{def-exp}
&\mathbb{E}\left[
\left|
{(1-\lambda)\over n}
\sum_{k=3}^n X_k
-
\left[1-{(1-\lambda)\over n}\right] X_1
+
\left[\lambda+{(1-\lambda)\over n}\right] X_2
\right|
\exp\left\{-z\left(\sum_{k=3}^n X_k
+X_1+X_2\right)\right\}
\right]
\nonumber
\\[0,2cm]
&=
\mathscr{L}_X^n(z)
\left\{
{(1-\lambda)\over n} \, 
\mathbb{E}\left[\left(\sum_{k=3}^n X_k\right)_z\right]
+
(\lambda+1)
\mathbb{E}[X_z]
\right\}
\nonumber
\\[0,2cm]
&
\times 
\pmb{G}\left({(1-\lambda)\over n}
\left(\sum_{k=3}^n X_k\right)_z+\left[\lambda+{(1-\lambda)\over n}\right] (X_2)_z, \left[{1-{(1-\lambda)\over n}}\right](X_1)_z\right).
	\end{align}		
%	{\small
%				\begin{align}\label{def-exp}
%			&\mathbb{E}\left[
%			\left|
%			{(1-\lambda)\over n}
%			\sum_{k=3}^n X_k
%			-
%			\left[1-{(1-\lambda)\over n}\right] X_1
%			+
%			\left[\lambda+{(1-\lambda)\over n}\right] X_2
%			\right|
%			\exp\left\{-z\left(\sum_{k=3}^n X_k
%			+X_1+X_2\right)\right\}
%			\right]
%			\nonumber
%			\\[0,2cm]
%			&=
%			\mathscr{L}_X^n(z)
%			\left\{
%			{(1-\lambda)\over n} \, 
%			\mathbb{E}\left[\left(\sum_{k=3}^n X_k\right)_z\right]
%			+
%			(\lambda+1)
%			\mathbb{E}[X_z]
%			-
%			2
%						\int_0^{\infty}
%			\overline{F}_{{(1-\lambda)\over n}
%				(\sum_{k=3}^n X_k)_z+\left[\lambda+{(1-\lambda)\over n}\right] (X_2)_z}(t^-)
%			\overline{F}_{\left[{1-{(1-\lambda)\over n}}\right]X_z}\left(t^-\right)
%			{\rm d}t
%			\right\}.
%			\end{align}
%		} \noindent
			Since $X_1,X_2,\ldots$ are i.i.d. with the same distribution as $X$, we have
			\begin{align}\label{pre-id-z}
				\left(\sum_{k=3}^n X_k\right)_z
				\stackrel{d}{=}
				\sum_{k=3}^n (X_k)_z,
				\quad 
			\mathbb{E}\left[\left(\sum_{k=3}^n X_k\right)_z\right]
			=
			(n-2)\mathbb{E}[X_z].
			\end{align}

Finally, the proof follows by combining \eqref{pre-id-0}, \eqref{def-exp}, and \eqref{pre-id-z}, along with the definition \eqref{def-I-lambda} of $\widehat I_\lambda$.
%
%Finally, using the representation
%\[
%G(Y)=\frac{\int_{0}^{\infty}F_Y(t^-)[1-F_Y(t^-)]\,{\rm d}t}{\mathbb{E}[Y]},
%\]
%of the Gini coefficient of a random variable $Y$, the proof of the theorem follows.
\end{proof}

\begin{corollary}\label{main-corollary-1}
	Under the assumptions of Theorem \ref{main-theorem}, it holds that
		\begin{align*}
		\mathbb{E}\left[\widehat{I}_\lambda\right]
		=
		\int_{0}^{\infty}
		\mathscr{L}_X^n(z)
		\left\{
		(n-1+\lambda)
		\mathbb{E}\left[X_z\right]
		-
		\int_0^{\infty}
		\overline{F}_{{(1-\lambda)}
			(\sum_{k=2}^n X_k)_z+n\lambda(X_2)_z}(t^-)
		\overline{F}_{(n-1+\lambda)X_z}\left(t^-\right)
		{\rm d}t
		\right\}
		{\rm d}z.
	\end{align*}
\end{corollary}
\begin{proof} For simplicity in the proof of this result, we adopt the notation for $W$, $Y$, $Z$, $a$, $b$, and $c$ introduced in \eqref{notations}.
	
	Using the identity \eqref{id-fund} together with the fact that $aW_z+bY_z$ and $cZ_z$ remain independent non–negative random variables, Tonelli's theorem yields
	\begin{align*}
		\mathbb{E}
		\left|
		aW_z+bY_z-cZ_z
		\right|
		=
		\int_0^{\infty}
		\left[
		\overline{F}_{aW_z+bY_z}(t^-)
		+
		\overline{F}_{cZ_z}\left(t^-\right)
		-
		2\,\overline{F}_{aW_z+bY_z}(t^-)
		\overline{F}_{cZ_z}\left(t^-\right)
		\right]
		{\rm d}t.
	\end{align*}
	
	Using 
	%the identity
	%$$
	%p+q-2pq=(1-p)+(1-q)-2(1-p)(1-q), \quad p,q\in[0,1],
	%$$
	%and 
	the well known representation \eqref{id-exp-posit},
	the previous identity can be written as
	\begin{align*}
		\mathbb{E}
		\left|
		aW_z+bY_z-cZ_z
		\right|
		=
		a\mathbb{E}[W_z]
		+
		c\mathbb{E}[Y_z]
		+
		b\mathbb{E}[Z_z]
		-
		2
		\int_0^{\infty}
		\overline{F}_{aW_z+bY_z}(t^-)
		\overline{F}_{cZ_z}\left(t^-\right)
		{\rm d}t.
	\end{align*}
	Hence,
	\begin{align}\label{iid-2}
		\pmb{G}(aW_z+bY_z,cZ_z)
		=
		{	a\mathbb{E}[W_z]
			+
			c\mathbb{E}[Y_z]
			+
			b\mathbb{E}[Z_z]
			-
			2
			\int_0^{\infty}
			\overline{F}_{aW_z+bY_z}(t^-)
			\overline{F}_{cZ_z}\left(t^-\right)
			{\rm d}t
			\over 
			a\mathbb{E}[W_z]
			+
			c\mathbb{E}[Y_z]
			+
			b\mathbb{E}[Z_z]}.
	\end{align}
	
	Combining Theorem \ref{main-theorem} with \eqref{iid-2} completes the proof.
\end{proof}

\begin{remark}
	It is important to note that, depending on the distribution specified for $X$, the Laplace transform in Theorem \ref{main-theorem} and Corollary \ref{main-corollary-1} may constrain the domain over which the improper integral is well-defined.
\end{remark}

\begin{remark}
	When $\lambda\to 0$, Corollary \ref{main-corollary-1} yields
	\begin{align*}
		\mathbb{E}\left[\widehat{H}\right]
		=
		\lim_{\lambda\to 0}
				\mathbb{E}\left[\widehat{I}_\lambda\right]
		=
		\int_{0}^{\infty}
		\mathscr{L}_X^n(z)
		\left\{
		(n-1)\mathbb{E}[X_z]
		-
		\int_0^{\infty}
		\overline{F}_{(\sum_{k=2}^n X_k)_z}(t^-)
		\overline{F}_{(n-1)X_z}(t^-)
		\,{\rm d}t
		\right\}
		{\rm d}z .
	\end{align*}
	An equivalent formula was recently obtained in Theorem~4.1 of \cite{Vila-Saulo2026}.
\end{remark}

\begin{remark}\label{rem-t1}
	Setting  $\lambda\to 1$ in Theorem \ref{main-theorem}, we obtain
		\begin{align*}
		\mathbb{E}\left[\widehat{G}\right]
=
		\lim_{\lambda\to 1}
\mathbb{E}\left[\widehat{I}_\lambda\right]
		=
		n
		\int_{0}^{\infty}
		\mathscr{L}_X^n(z)
		\mathbb{E}[X_z]
		\pmb{G}((X_2)_z, (X_1)_z)
		{\rm d}z.
	\end{align*}
Since $(X_1)_z$ and $(X_2)_z$ are i.i.d., because $X_1$ and $X_2$ are i.i.d. copies of $X$, the normalized mean absolute difference $\pmb{G}((X_2)_z, (X_1)_z)$ reduces to the Gini coefficient $G(X_z)$. Therefore, the above identity can be written as
	\begin{align*}
		\mathbb{E}\left[\widehat{G}\right]
		=
		n
		\int_{0}^{\infty}
		\mathscr{L}_X^n(z)
		\mathbb{E}[X_z]
		G(X_z)
		\,{\rm d}z .
	\end{align*}
	This representation has recently appeared in Theorem~3.2 of \cite{Vila-Saulo2025}.
\end{remark}

\begin{remark}
As an application of Corollary \ref{main-corollary-1}, in the next section we derive a simple expression for $	\mathbb{E}\left[\widehat{I}_\lambda\right]$ in the case of gamma distributions (including the Exponential, Erlang, Chi-square, and Scaled Gamma distributions). The gamma distribution is a natural choice due to its flexibility and the fact that its tilted versions, sums, and differences lead to distributions that are well known in the literature.
\end{remark}

\section{Bias of the estimator $\widehat{I}_\lambda$ under gamma populations}\label{Bias of the estimator}

Let $X_1,X_2,\ldots$ be i.i.d. random variables with
$
X\sim \mathrm{Gamma}(\alpha,\beta),
$
where $\alpha>0$ is the shape parameter and $\beta>0$ is the rate parameter. 
The Laplace transform of $X$ is given by
\[
\mathscr{L}_X(z)
=
\left(\frac{\beta}{\beta+z}\right)^{\alpha}, \quad z>-\beta.
\]

Under exponential tilting with parameter $z>0$, it is well known that
\begin{align}\label{dist-gamma}
	X_z \sim \mathrm{Gamma}(\alpha,\beta+z).
\end{align}
%In particular,
%$
%	\mathbb{E}\left[X_z\right]={\alpha/(\beta+z)}.
%$

\smallskip 
Unless otherwise specified, in the remainder of this section we assume that $\lambda\in[0,1)$.

From \eqref{dist-gamma} and the scale property of the gamma distribution, we obtain
\[
(n-1+\lambda)X_z \sim \mathrm{Gamma}\left(\alpha,{\beta+z\over n-1+\lambda}\right).
\]

Moreover, noting that
\[
\sum_{k=2}^n X_k = \sum_{k=3}^n X_k+X_2,
\]
we obtain
\begin{align}\label{gamma-1}
(1-\lambda)
\left(\sum_{k=2}^n X_k\right)_z+n\lambda(X_2)_z
=
(1-\lambda)\left(\sum_{k=3}^n X_k\right)_z 
+
[1+(n-1)\lambda](X_2)_z.
\end{align}

Since $X_1,X_2,\ldots$ are independent,
\[
\left(\sum_{k=3}^n X_k\right)_z \sim \mathrm{Gamma}((n-2)\alpha,\beta+z),
\quad 
(X_2)_z \sim \mathrm{Gamma}(\alpha,\beta+z).
\]

Hence,
\begin{align}\label{gamma-2}
(1-\lambda)
\left(\sum_{k=3}^n X_k\right)_z \sim \mathrm{Gamma}\left((n-2)\alpha,{\beta+z\over 1-\lambda}\right),
\quad 
[1+(n-1)\lambda](X_2)_z \sim \mathrm{Gamma}\left(\alpha,{\beta+z\over 1+(n-1)\lambda}\right).
\end{align}

\bigskip 
\begin{definition}\label{def-main}
Let $X_1,X_2$ be independent random variables such that
$
X_i \sim \mathrm{Gamma}(\alpha_i,\beta_i), \ \alpha_i>0,\; \beta_i>0,\; i=1,2,
$
where $\alpha_i$ and $\beta_i$ denote the shape and rate parameters, respectively. 
Then the random variable
$
S=X_1+X_2
$
is said to follow a \textbf{generalized hypoexponential distribution} (GHypo) \citep{Moschopoulos1985}, and we write
$
S \sim \textrm{GHypo}(\alpha_1,\beta_1,\alpha_2,\beta_2).
$
	Its PDF is given by 
	\[
	f_{\text{GHypo}(\alpha_1,\beta_1,\alpha_2,\beta_2)}(t)
	=
	\frac{\beta_1^{\alpha_1}\beta_2^{\alpha_2}}
	{\Gamma(\alpha_1+\alpha_2)} \
	t^{\alpha_1+\alpha_2-1}
	\exp\{-\beta_2 t\}
	\,{}_1F_1
	\!\left(\alpha_1;\alpha_1+\alpha_2;(\beta_2-\beta_1)t\right),
	\quad t>0,
	\]
	where ${}_1F_1$ denotes the confluent hypergeometric function of Kummer. Its corresponding CDF is denoted by 
	$F_{\text{GHypo}(\alpha_1,\beta_1,\alpha_2,\beta_2)}$ and can be obtained by using the identity
	\[
	\int_0^t
	x^{\nu-1} \exp\{-\mu x\}
	\,{}_1F_1(a;\nu;cx)\,{\rm d}x
	=
	\frac{t^{\nu}}{\nu} \, 
	\exp\{-\mu t\}
	\,{}_1F_1\!\left(a;\nu+1;(c-\mu)t\right),
	\quad \nu>0,
	\]
	yielding
	\begin{align}\label{CDF-GHypo}
F_{\mathrm{GHypo}(\alpha_1,\beta_1,\alpha_2,\beta_2)}(t)
=
\frac{\beta_1^{\alpha_1}\beta_2^{\alpha_2}}
{\Gamma(\alpha_1+\alpha_2+1)} \,
t^{\alpha_1+\alpha_2}
\exp\{-\beta_2 t\}
\,{}_1F_1
\!\left(
\alpha_1;\alpha_1+\alpha_2+1;(\beta_2-\beta_1)t
\right),
\quad t>0.
	\end{align}
%	\[
%	\mathscr{L}_S(s)
%	=
%	\prod_{i=1}^m
%	\left(\frac{\beta_i}{\beta_i+s}\right)^{\alpha_i},
%	\qquad s>0.
%	\]
\end{definition}

\bigskip 
Combining Definition \ref{def-main} with \eqref{gamma-1} and \eqref{gamma-2}, we obtain
\begin{align*}
	(1-\lambda)
	\left(\sum_{k=2}^n X_k\right)_z+n\lambda(X_2)_z\sim 
	\text{GHypo}\left((n-2)\alpha,{\beta+z\over 1-\lambda},\alpha,{\beta+z\over 1+(n-1)\lambda}\right).
\end{align*}

Applying Corollary \ref{main-corollary-1}, for $\lambda\in[0,1)$, we obtain
	\begin{align}\label{exp-gamma}
	\mathbb{E}\left[\widehat{I}_\lambda\right]
	&=
	\int_{0}^{\infty}
	\left(\frac{\beta}{\beta+z}\right)^{n\alpha}
	\left\{ 
	{(n-1+\lambda)\alpha\over \beta+z}
	-
	\int_0^{\infty}
	\overline{F}_{	\text{GHypo}\left((n-2)\alpha,{\beta+z\over 1-\lambda},\alpha,{\beta+z\over 1+(n-1)\lambda}\right)}(t) \, 
	{\Gamma\left(\alpha,{(\beta+z)t\over n-1+\lambda} \right)\over\Gamma(\alpha)} \, 
	{\rm d}t
	\right\}
	{\rm d}z
	\nonumber
	\\[0,2cm]
	&=
	\int_{0}^{\infty}
	\left(\frac{\beta}{\beta+z}\right)^{n\alpha}
	\left\{ 
	{(n-1+\lambda)\alpha\over \beta+z}
	-
	\int_0^{\infty}
	\overline{F}_{	\text{GHypo}\left((n-2)\alpha,{1\over 1-\lambda},\alpha,{1\over 1+(n-1)\lambda}\right)}((\beta+z)t) \, 
	{\Gamma\left(\alpha,{(\beta+z)t\over n-1+\lambda} \right)\over\Gamma(\alpha)} \,
	{\rm d}t
	\right\}
	{\rm d}z
	\nonumber
		\\[0,2cm]
	&=
%	\int_{0}^{\infty}
%	\frac{\beta^{n\alpha}}{(\beta+z)^{n\alpha+1}}
%	{\rm d}z
\frac{1}{\alpha}
	\left\{ 
	{\left(1+{\lambda-1\over n}\right)\alpha}
	-
	{1\over n}
	\int_0^{\infty}
	\overline{F}_{	\text{GHypo}\left((n-2)\alpha,{1\over 1-\lambda},\alpha,{1\over 1+(n-1)\lambda}\right)}(t) \, 
	{\Gamma\left(\alpha,{t\over n-1+\lambda} \right)\over\Gamma(\alpha)} \,
	{\rm d}t
	\right\},
\end{align}
where in the second equality we used the scale property of the GHypo distribution, and in the third equality the change of variable $t \mapsto (\beta+z)t$ was performed.

\begin{remark}
Because the estimator $\widehat{I}_\lambda$ is scale invariant, its expectation in Item \eqref{exp-gamma} does not depend on the rate parameter $\beta$.
\end{remark}

Therefore, by combining Proposition \ref{prop-Hoover-index} and  \eqref{exp-gamma}, we have:
\begin{corollary}\label{main-corollary}
For $n\geqslant 2$ and $\lambda\in(0,1)$, the bias of $\widehat{I}_\lambda$ relative to $I_\lambda$, denoted by $\text{Bias}(\widehat{I}_\lambda,I_\lambda)$, is given by
	\begin{multline*}
\text{Bias}(\widehat{I}_\lambda,I_\lambda)
		=
\frac{1}{\alpha}
\left\{ 
{\left(1+{\lambda-1\over n}\right)\alpha}
-
{1\over n}
\int_0^{\infty}
\overline{F}_{	\text{GHypo}\left((n-2)\alpha,{1\over 1-\lambda},\alpha,{1\over 1+(n-1)\lambda}\right)}(t) \, 
{\Gamma\left(\alpha,{t\over n-1+\lambda} \right)\over\Gamma(\alpha)} \,
{\rm d}t
\right\}
\\[0,2cm]
		-
		\left\{
\frac{(1-\lambda)^\alpha\alpha^{\alpha-1} \exp\{-(1-\lambda)\alpha\} }{\Gamma(\alpha)}
+
\lambda\, \frac{\Gamma(\alpha,(1-\lambda)\alpha)}{\Gamma(\alpha)}
-
{1\over \alpha}
\int_{(1-\lambda)\alpha}^\infty
{\Gamma(\alpha, t)\over\Gamma(\alpha)} \,
{\Gamma\left(\alpha,{ t-(1-\lambda)\alpha\over\lambda}\right)\over\Gamma(\alpha)}	\,
{\rm d}t
		\right\},
	\end{multline*}
or, equivalently, by using \eqref{CDF-GHypo},
		\begin{multline*}
		\text{Bias}(\widehat{I}_\lambda,I_\lambda)
		=
\frac{
	\left(\frac{1}{1-\lambda}\right)^{(n-2)\alpha}
	\left[\frac{1}{1+(n-1)\lambda}\right]^{\alpha}
}{
	n\Gamma((n-1)\alpha+1)\Gamma(\alpha+1)
} 
		\\[0,2cm]\times 
\int_0^\infty
t^{(n-1)\alpha}
\exp\left\{-\frac{t}{1+(n-1)\lambda}\right\}
\Gamma\left(\alpha,{t\over n-1+\lambda} \right)
\,
{}_1F_1
\!\left(
(n-2)\alpha;
(n-1)\alpha+1;
-\frac{n\lambda t}{(1-\lambda)(1+(n-1)\lambda)}
\right)
{\rm d}t
		\\[0,2cm]
		-
		\left\{
	\frac{(1-\lambda)^\alpha\alpha^{\alpha-1} \exp\{-(1-\lambda)\alpha\} }{\Gamma(\alpha)}
	+
	\lambda\, \frac{\Gamma(\alpha,(1-\lambda)\alpha)}{\Gamma(\alpha)}
	-
	{1\over \alpha}
	\int_{(1-\lambda)\alpha}^\infty
	{\Gamma(\alpha, t)\over\Gamma(\alpha)} \,
	{\Gamma\left(\alpha,{ t-(1-\lambda)\alpha\over\lambda}\right)\over\Gamma(\alpha)}	\,
	{\rm d}t
		\right\}.
	\end{multline*}
\end{corollary}

\begin{remark}
	Since 
	\[
	\lim_{\lambda\to 0^+}
	F_{	\text{GHypo}\left((n-2)\alpha,{1\over 1-\lambda},\alpha,{1\over 1+(n-1)\lambda}\right)}(t)
	=
	F_{
	\mathrm{GHypo}\big((n-2)\alpha,1,\alpha,1\big)}(t)
	=
	F_{
	\mathrm{Gamma}((n-1)\alpha,1)}(t),
	\]
	combining Corollary \ref{main-corollary} with Remark \ref{rem-H} yields
	\begin{multline*}
		\text{Bias}(\widehat{H},H)
		=
			\lim_{\lambda\to 0^+}
\text{Bias}(\widehat{I}_\lambda,I_\lambda)
		\\[0,2cm]
		=
		\frac{1}{\alpha}
		\left\{
		\left(1-\frac{1}{n}\right)\alpha
		-
		\frac{1}{n}
		\int_0^{\infty}
		\frac{\Gamma\left((n-1)\alpha,t\right)}{\Gamma((n-1)\alpha)}
		\,
		\frac{\Gamma\left(\alpha,\frac{t}{n-1}\right)}{\Gamma(\alpha)}
		\,\mathrm{d}t
		\right\}
		-
		\frac{\alpha^{\alpha-1}\exp\{-\alpha\}}{\Gamma(\alpha)}.
	\end{multline*}
	Note that the above identity recently appeared in Corollary 4.5 of \cite{Vila-Saulo2026}.
\end{remark}

\begin{remark}
	Since
	\[
	\lim_{\lambda\to 1^-}
	F_{\text{GHypo}\left((n-2)\alpha,\frac{1}{1-\lambda},\alpha,\frac{1}{1+(n-1)\lambda}\right)}(t)
	=
	F_{{\rm Gamma}(\alpha,{1\over n})}(t),
	\]
	combining Corollary \ref{main-corollary} with Remark \ref{rem-G} gives
	\begin{align*}
	\text{Bias}(\widehat{G},G)
	=
	\lim_{\lambda\to 1^-}
	\text{Bias}(\widehat{I}_\lambda,I_\lambda)
	&=
	\frac{1}{\alpha}
	\left\{ 
	\alpha
	-
	{1\over n}
	\int_0^{\infty}
	\frac{\Gamma^2\left(\alpha,\frac{t}{n}\right)}{\Gamma^2(\alpha)} \, 
	{\rm d}t
	\right\}
	-
	\frac{\Gamma\left(\alpha+\tfrac12\right)}
{\sqrt{\pi}\alpha\Gamma(\alpha)}
\\[0,2cm]
&=
1
-
{1\over \alpha}
\int_0^{\infty}
\frac{\Gamma^2\left(\alpha,t\right)}{\Gamma^2(\alpha)} \, 
{\rm d}t
-
\frac{\Gamma\left(\alpha+\tfrac12\right)}
{\sqrt{\pi}\alpha\Gamma(\alpha)}
=
0,
	\end{align*}
where the third equality follows from the change of variables $t \mapsto t/n$, 
and the last equality follows from identity \eqref{ide-last}.

The above result on the unbiasedness of the Gini coefficient estimator $\widehat{G}$ for gamma populations is well established in the literature; see, for example, \cite{Baydil2025,Vila-Saulo2025}.
\end{remark}

\section{Monte Carlo simulation study}\label{Monte Carlo simulation study}

In this section, we assess the finite-sample performance of the plug-in estimator $\widehat I_\lambda$ defined in Equation \eqref{def-I-lambda} by means of a Monte Carlo simulation study. The analysis is conducted under gamma populations. Let $X_1,\ldots,X_n$ be a random sample from a gamma distribution with shape parameter $\alpha>0$ and rate parameter $\beta>0$. Since the proposed index $I_\lambda$ is scale invariant, its value does not depend on $\beta$. Accordingly, without loss of generality, we set $\beta=1$ throughout the simulation study and vary only the shape parameter $\alpha$, the interpolation parameter $\lambda$, and the sample size $n$. This choice allows us to isolate the effects of distributional asymmetry and sample size on the estimator's performance.

\subsection{Evaluating the Bias and MSE of estimates}

The simulation design considers the values
\[
\alpha \in \{0.5,\,1,\,2,\,5,\,10\}, \quad
\lambda \in \{0.25,\,0.50,\,0.75\}, \quad
n \in \{10,\,20,\,40,\,80,\,120\}.
\]
These settings cover a broad range of relevant scenarios. In particular, smaller values of $\alpha$ correspond to more skewed gamma distributions, whereas larger values of $\alpha$ produce less dispersed and more symmetric populations. Likewise, the selected values of $\lambda$ represent distinct points in the interior of the family bridging the Hoover and Gini measures. For each combination of $(\alpha,\lambda,n)$, we generate $R=1000$ independent Monte Carlo samples.

For each replication, a sample $X_1,\ldots,X_n$ is generated from the $\mathrm{Gamma}(\alpha,1)$ distribution, and the estimator $\widehat I_\lambda$ is computed according to Equation \eqref{def-I-lambda}. The corresponding population value $I_\lambda$ is evaluated numerically using the integral representation derived for gamma distributions in Proposition \ref{prop-Hoover-index}. For each scenario, the true value $I_\lambda$ is computed once and used as a benchmark across all replications.

The performance of the estimator is evaluated using the empirical bias, mean squared error (MSE), and variance. Specifically, we compute
\[
\operatorname{Bias}_{MC}(\widehat I_\lambda)
=
\frac{1}{R}\sum_{r=1}^{R}\bigl(\widehat I_\lambda^{(r)}-I_\lambda\bigr),
\]
\[
\operatorname{MSE}_{MC}(\widehat I_\lambda)
=
\frac{1}{R}\sum_{r=1}^{R}\bigl(\widehat I_\lambda^{(r)}-I_\lambda\bigr)^2,
\]
and
\[
\operatorname{Var}_{MC}(\widehat I_\lambda)
=
\frac{1}{R-1}\sum_{r=1}^{R}
\bigl(\widehat I_\lambda^{(r)}-\overline{\widehat I_\lambda}\bigr)^2,
\]
where $\overline{\widehat I_\lambda} = ({1}/{R})\sum_{r=1}^R \widehat I_\lambda^{(r)}$ denotes the Monte Carlo average of the estimator.
These measures allow us to separately assess the systematic error (bias), the overall accuracy (MSE), and the dispersion of the estimator (variance). 
%In particular, recall that
% \[
% \operatorname{MSE}(\hat I_\lambda)
% =
% \operatorname{Bias}^2(\hat I_\lambda) + \operatorname{Var}(\hat I_\lambda),
% \]
% which highlights the trade-off between bias and variability.

The numerical implementation was carried out in \textsf{R}. The integral expressions required to evaluate $I_\lambda$ were computed numerically using adaptive quadrature methods. For each parameter configuration, the Monte Carlo summaries were obtained from $R=1000$ independent replications. The purpose of this study is threefold. First, it allows us to quantify the magnitude and direction of the finite-sample bias of $\widehat I_\lambda$ under different levels of skewness and sample size. Second, it evaluates the overall accuracy of the estimator through the MSE. Third, the empirical variance provides insight into the variability of the estimator across repeated samples.

The Monte Carlo results, reported in Table 1, indicate that the estimator exhibits the expected improvement as the sample size increases. In particular, both the MSE and the variance decrease systematically with $n$, indicating increased precision. The bias also tends to decrease in magnitude for moderate and large samples. Moreover, the results suggest that the estimator's performance depends on the shape parameter $\alpha$, with more skewed distributions (small $\alpha$) generally leading to larger bias and variability. This behavior is consistent with the increased difficulty of estimating inequality measures in highly asymmetric populations.

\begin{center}
\begin{longtable}{rrrrrrrr}
\caption{Monte Carlo results for the estimator $\widehat{I}_\lambda$ under Gamma$(\alpha,1)$ distributions, based on $R=1000$ replications. The table reports the true value $I_\lambda$, the Monte Carlo mean (Mean), empirical bias (Bias), mean squared error (MSE), and variance (Var).} \\
\toprule
$\alpha$ & $\lambda$ & $n$ & $I_{\lambda}$ & Mean & Bias & MSE & Var \\
\midrule
\endfirsthead

\toprule
$\alpha$ & $\lambda$ & $n$ & $I_{\lambda}$ & Mean & Bias & MSE & Var \\
\midrule
\endhead

\midrule
\multicolumn{8}{r}{\textit{Continued on next page}} \\
\midrule
\endfoot

\bottomrule
\endlastfoot
0.5 & 0.25 & 10 & 0.4959 & 0.4804 & -0.0156 & 0.0085 & 0.0082 \\
1.0 & 0.25 & 10 & 0.3779 & 0.3597 & -0.0182 & 0.0062 & 0.0059 \\
2.0 & 0.25 & 10 & 0.2785 & 0.2738 & -0.0047 & 0.0040 & 0.0040 \\
5.0 & 0.25 & 10 & 0.1807 & 0.1750 & -0.0057 & 0.0018 & 0.0017 \\
10.0 & 0.25 & 10 & 0.1289 & 0.1251 & -0.0038 & 0.0009 & 0.0009 \\
0.5 & 0.50 & 10 & 0.5260 & 0.5256 & -0.0004 & 0.0087 & 0.0087 \\
1.0 & 0.50 & 10 & 0.4044 & 0.3991 & -0.0053 & 0.0064 & 0.0064 \\
2.0 & 0.50 & 10 & 0.2998 & 0.2949 & -0.0049 & 0.0042 & 0.0042 \\
5.0 & 0.50 & 10 & 0.1954 & 0.1923 & -0.0031 & 0.0019 & 0.0019 \\
10.0 & 0.50 & 10 & 0.1396 & 0.1396 & 0.0000 & 0.0011 & 0.0011 \\
0.5 & 0.75 & 10 & 0.5718 & 0.5723 & 0.0005 & 0.0091 & 0.0091 \\
1.0 & 0.75 & 10 & 0.4450 & 0.4432 & -0.0018 & 0.0073 & 0.0073 \\
2.0 & 0.75 & 10 & 0.3324 & 0.3332 & 0.0007 & 0.0052 & 0.0052 \\
5.0 & 0.75 & 10 & 0.2177 & 0.2163 & -0.0014 & 0.0026 & 0.0026 \\
10.0 & 0.75 & 10 & 0.1558 & 0.1549 & -0.0009 & 0.0013 & 0.0013 \\

0.5 & 0.25 & 20 & 0.4959 & 0.4882 & -0.0077 & 0.0042 & 0.0042 \\
1.0 & 0.25 & 20 & 0.3779 & 0.3709 & -0.0070 & 0.0030 & 0.0029 \\
2.0 & 0.25 & 20 & 0.2785 & 0.2749 & -0.0036 & 0.0019 & 0.0019 \\
5.0 & 0.25 & 20 & 0.1807 & 0.1794 & -0.0014 & 0.0009 & 0.0009 \\
10.0 & 0.25 & 20 & 0.1289 & 0.1269 & -0.0020 & 0.0005 & 0.0005 \\
0.5 & 0.50 & 20 & 0.5260 & 0.5233 & -0.0027 & 0.0040 & 0.0040 \\
1.0 & 0.50 & 20 & 0.4044 & 0.4014 & -0.0030 & 0.0031 & 0.0031 \\
2.0 & 0.50 & 20 & 0.2998 & 0.2998 & 0.0000 & 0.0020 & 0.0020 \\
5.0 & 0.50 & 20 & 0.1954 & 0.1970 & 0.0016 & 0.0010 & 0.0010 \\
10.0 & 0.50 & 20 & 0.1396 & 0.1382 & -0.0014 & 0.0005 & 0.0005 \\
0.5 & 0.75 & 20 & 0.5718 & 0.5711 & -0.0007 & 0.0045 & 0.0045 \\
1.0 & 0.75 & 20 & 0.4450 & 0.4455 & 0.0005 & 0.0036 & 0.0036 \\
2.0 & 0.75 & 20 & 0.3324 & 0.3296 & -0.0028 & 0.0024 & 0.0024 \\
5.0 & 0.75 & 20 & 0.2177 & 0.2174 & -0.0003 & 0.0012 & 0.0012 \\
10.0 & 0.75 & 20 & 0.1558 & 0.1557 & -0.0001 & 0.0006 & 0.0006 \\

0.5 & 0.25 & 40 & 0.4959 & 0.4946 & -0.0013 & 0.0021 & 0.0021 \\
1.0 & 0.25 & 40 & 0.3779 & 0.3756 & -0.0023 & 0.0014 & 0.0014 \\
2.0 & 0.25 & 40 & 0.2785 & 0.2772 & -0.0013 & 0.0009 & 0.0009 \\
5.0 & 0.25 & 40 & 0.1807 & 0.1802 & -0.0005 & 0.0004 & 0.0004 \\
10.0 & 0.25 & 40 & 0.1289 & 0.1281 & -0.0008 & 0.0002 & 0.0002 \\
0.5 & 0.50 & 40 & 0.5260 & 0.5236 & -0.0024 & 0.0021 & 0.0021 \\
1.0 & 0.50 & 40 & 0.4044 & 0.4038 & -0.0005 & 0.0016 & 0.0016 \\
2.0 & 0.50 & 40 & 0.2998 & 0.3001 & 0.0003 & 0.0010 & 0.0010 \\
5.0 & 0.50 & 40 & 0.1954 & 0.1947 & -0.0007 & 0.0005 & 0.0005 \\
10.0 & 0.50 & 40 & 0.1396 & 0.1393 & -0.0003 & 0.0003 & 0.0003 \\
0.5 & 0.75 & 40 & 0.5718 & 0.5681 & -0.0037 & 0.0022 & 0.0022 \\
1.0 & 0.75 & 40 & 0.4450 & 0.4432 & -0.0019 & 0.0018 & 0.0018 \\
2.0 & 0.75 & 40 & 0.3324 & 0.3311 & -0.0014 & 0.0011 & 0.0011 \\
5.0 & 0.75 & 40 & 0.2177 & 0.2180 & 0.0003 & 0.0006 & 0.0006 \\
10.0 & 0.75 & 40 & 0.1558 & 0.1558 & -0.0001 & 0.0003 & 0.0003 \\

0.5 & 0.25 & 80 & 0.4959 & 0.4937 & -0.0023 & 0.0010 & 0.0010 \\
1.0 & 0.25 & 80 & 0.3779 & 0.3762 & -0.0017 & 0.0007 & 0.0007 \\
2.0 & 0.25 & 80 & 0.2785 & 0.2767 & -0.0018 & 0.0005 & 0.0005 \\
5.0 & 0.25 & 80 & 0.1807 & 0.1796 & -0.0012 & 0.0002 & 0.0002 \\
10.0 & 0.25 & 80 & 0.1289 & 0.1283 & -0.0006 & 0.0001 & 0.0001 \\
0.5 & 0.50 & 80 & 0.5260 & 0.5254 & -0.0007 & 0.0012 & 0.0012 \\
1.0 & 0.50 & 80 & 0.4044 & 0.4037 & -0.0007 & 0.0008 & 0.0008 \\
2.0 & 0.50 & 80 & 0.2998 & 0.2988 & -0.0010 & 0.0005 & 0.0005 \\
5.0 & 0.50 & 80 & 0.1954 & 0.1956 & 0.0002 & 0.0002 & 0.0002 \\
10.0 & 0.50 & 80 & 0.1396 & 0.1392 & -0.0004 & 0.0001 & 0.0001 \\
0.5 & 0.75 & 80 & 0.5718 & 0.5701 & -0.0017 & 0.0010 & 0.0010 \\
1.0 & 0.75 & 80 & 0.4450 & 0.4452 & 0.0001 & 0.0009 & 0.0009 \\
2.0 & 0.75 & 80 & 0.3324 & 0.3324 & -0.0000 & 0.0006 & 0.0006 \\
5.0 & 0.75 & 80 & 0.2177 & 0.2177 & -0.0001 & 0.0003 & 0.0003 \\
10.0 & 0.75 & 80 & 0.1558 & 0.1557 & -0.0001 & 0.0002 & 0.0002 \\

0.5 & 0.25 & 120 & 0.4959 & 0.4946 & -0.0013 & 0.0007 & 0.0007 \\
1.0 & 0.25 & 120 & 0.3779 & 0.3774 & -0.0005 & 0.0005 & 0.0005 \\
2.0 & 0.25 & 120 & 0.2785 & 0.2778 & -0.0007 & 0.0003 & 0.0003 \\
5.0 & 0.25 & 120 & 0.1807 & 0.1799 & -0.0009 & 0.0001 & 0.0001 \\
10.0 & 0.25 & 120 & 0.1289 & 0.1285 & -0.0004 & 0.0001 & 0.0001 \\
0.5 & 0.50 & 120 & 0.5260 & 0.5254 & -0.0007 & 0.0007 & 0.0007 \\
1.0 & 0.50 & 120 & 0.4044 & 0.4048 & 0.0005 & 0.0005 & 0.0005 \\
2.0 & 0.50 & 120 & 0.2998 & 0.2995 & -0.0003 & 0.0004 & 0.0004 \\
5.0 & 0.50 & 120 & 0.1954 & 0.1958 & 0.0004 & 0.0001 & 0.0001 \\
10.0 & 0.50 & 120 & 0.1396 & 0.1393 & -0.0003 & 0.0001 & 0.0001 \\
0.5 & 0.75 & 120 & 0.5718 & 0.5715 & -0.0003 & 0.0007 & 0.0007 \\
1.0 & 0.75 & 120 & 0.4450 & 0.4454 & 0.0004 & 0.0005 & 0.0005 \\
2.0 & 0.75 & 120 & 0.3324 & 0.3326 & 0.0002 & 0.0004 & 0.0004 \\
5.0 & 0.75 & 120 & 0.2177 & 0.2177 & -0.0000 & 0.0002 & 0.0002 \\
10.0 & 0.75 & 120 & 0.1558 & 0.1558 & -0.0001 & 0.0001 & 0.0001 \\

\end{longtable}\label{table:mc}
\end{center}

\subsection{Comparison with the competing index $J_\lambda$}

In this subsection, we compare the bias of the estimates for $I_\lambda$ with the competing index $J_\lambda$, defined respectively in \eqref{eq:index} and \eqref{ineq-H-G}.

The natural estimator for \( J_\lambda \), as defined in \eqref{ineq-H-G-1}, is
$$\widehat{J}_\lambda = (1-\lambda) \widehat{H} + \lambda \widehat{G},$$
where $\widehat{H}$ and $\widehat{G}$ are given in \eqref{eq:Hoover_estimator} and \eqref{eq:Gini_estimator}, respectively.

Fixing $\alpha=5$ and varying $\lambda$ and $n$, Figure~\ref{fig:comparison} shows that, in general, the bias values of $\widehat{I}_\lambda$ are smaller than those of the competing index, as previously reported in Remark \ref{remark:J_index}.
This shows that although $J_\lambda$ has a simpler mathematical formulation, its estimation does not improve upon the good results of the $\widehat {I}_{\lambda }$ estimator.

Recall that $I_{\lambda }\leqslant J_{\lambda }$ and $\widehat I_{\lambda }\leqslant \widehat J_{\lambda }$ (see Remarks \ref{remark:J_index} and \ref{remark:J_index-1}). A relative bias may be useful to analyze, as it avoids the influence of the true parameter values in these comparisons. However, the graphical analyses obtained were essentially the same, differing only by the scale of the y-axis. Therefore, we present only the biased results, since our conclusions would remain unchanged under the relative bias evaluation metric.

\begin{figure}[ht]
\centering
\includegraphics[width=0.90\textwidth]{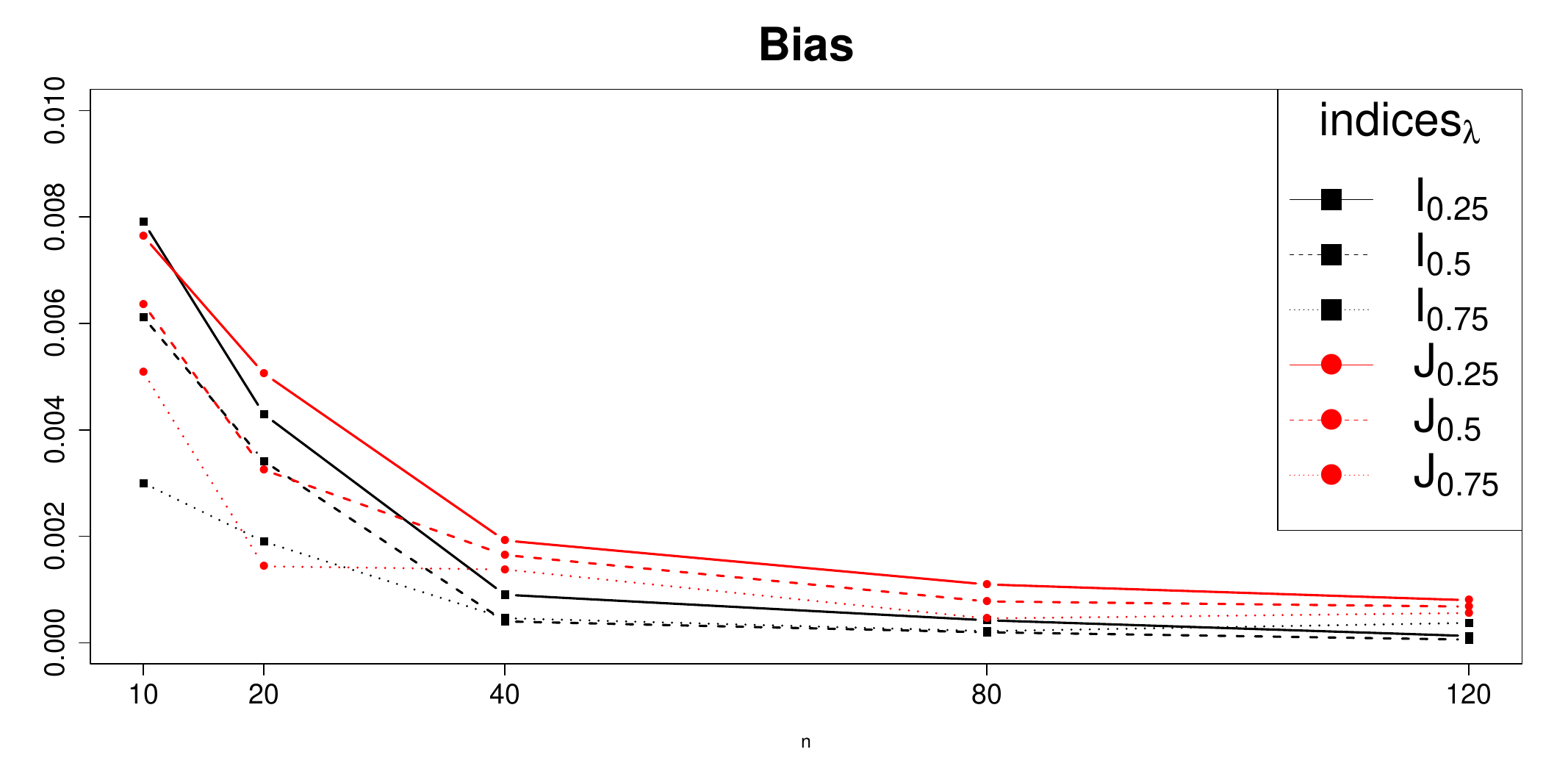}
\caption{Bias of the estimators $\widehat{I}_\lambda$ (black) and $\widehat{J}_\lambda$ (red) under Gamma$(\alpha,1)$ distributions, based on $R=1000$ Monte Carlo replications, for $\lambda\in\{0.25, 0.5, 0.75\}$.}
\label{fig:comparison}
\end{figure}

\section{Empirical application}\label{empirical application}

To illustrate the practical usefulness of the proposed family of inequality measures, we consider the distribution of GDP per capita across countries in the Americas using data from the World Bank. The variable of interest is GDP per capita measured at purchasing power parity (PPP), expressed in constant 2021 international dollars; see \url{https://ourworldindata.org/grapher/gdp-per-capita-worldbank}. Since the proposed index is scale invariant, the units of measurement do not affect the results.

We compute the Hoover index, the Gini coefficient, and the proposed estimator $\widehat I_\lambda$ for different values of $\lambda \in [0,1]$. Recall that, by construction, the proposed index satisfies $I_0 = H$ and $I_1 = G$, so that intermediate values of $\lambda$ provide a continuous transition between these two classical measures. Table~\ref{tab:application_indices} reports the values of the Hoover index, the Gini coefficient, and selected values of the proposed index. The estimated Hoover index is 0.229, while the Gini coefficient is given by 0.329. The intermediate values $\widehat I_{0.25}$, $\widehat I_{0.50}$, and $\widehat I_{0.75}$ lie between these two benchmarks, as expected.

\begin{table}[ht]
\centering
\begin{tabular}{lc}
\toprule
Measure & Value \\
\midrule
Hoover & 0.229 \\
$I_{0.25}$ & 0.237 \\
$I_{0.50}$ & 0.258 \\
$I_{0.75}$ & 0.290 \\
Gini & 0.329 \\
\bottomrule
\end{tabular}
\caption{Inequality measures for GDP per capita across countries in the Americas.}
\label{tab:application_indices}
\end{table}

Figure~\ref{fig:path} displays the function $\lambda \mapsto \widehat I_\lambda$. As expected from the theoretical construction, the curve is continuous and increasing, connecting the Hoover index at $\lambda=0$ to the Gini coefficient at $\lambda=1$. This confirms that the proposed index provides a smooth interpolation between these two measures.
\begin{figure}[ht]
\centering
\includegraphics[width=0.90\textwidth]{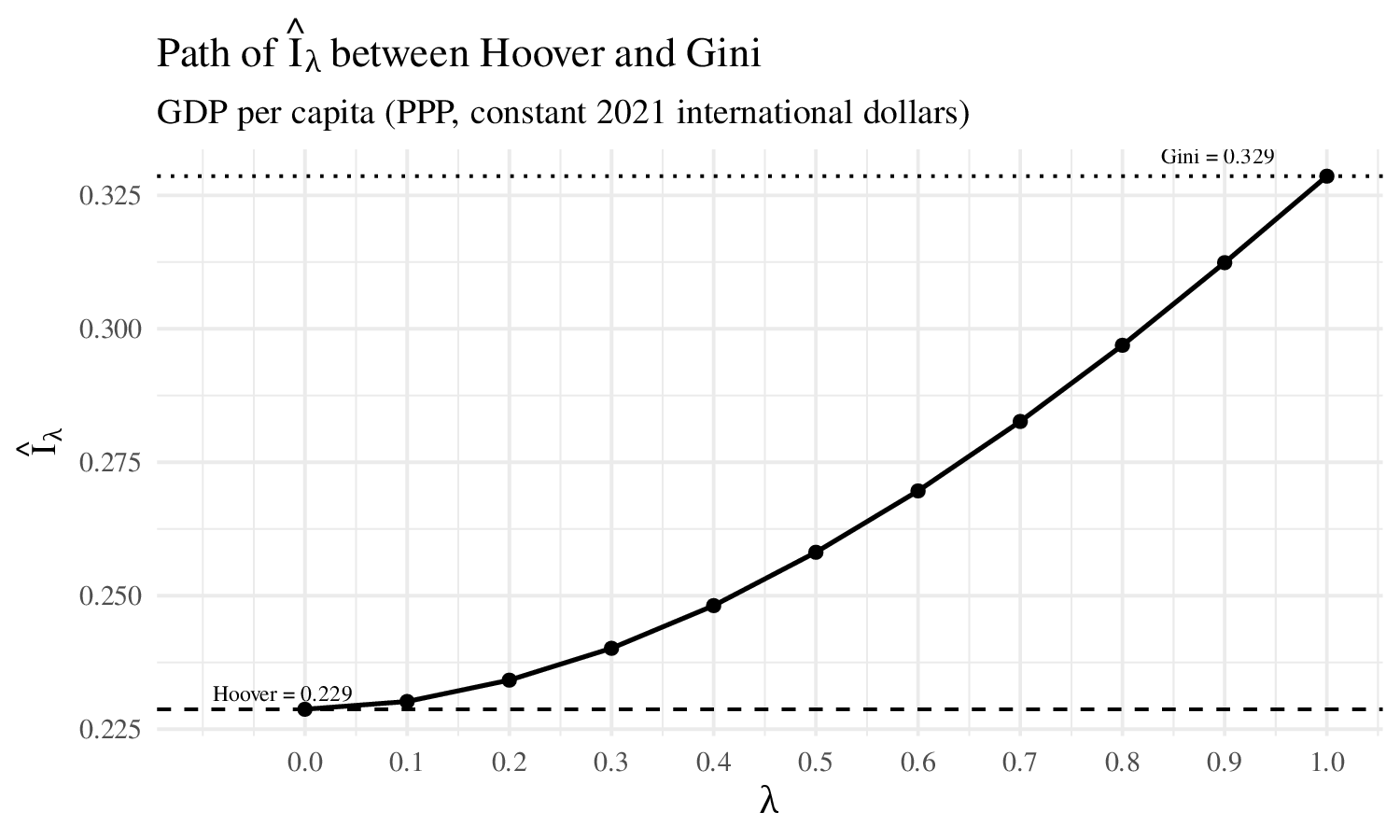}
\caption{Path of $\widehat I_\lambda$ as a function of $\lambda$, with Hoover and Gini as limiting cases.}
\label{fig:path}
\end{figure}

The empirical results highlight the flexibility of the proposed index. For values of $\lambda$ close to zero, the index behaves similarly to the Hoover measure, placing more emphasis on deviations from the mean. In contrast, for values of $\lambda$ close to one, the index approaches the Gini coefficient and therefore reflects pairwise differences across countries. This behavior can be clearly observed in Figure~\ref{fig:path}, where the slope of the curve indicates how sensitive the inequality assessment is to the choice of $\lambda$. A relatively smooth transition suggests that the overall level of inequality is robust to the choice of measure, whereas steeper changes would indicate greater sensitivity.

\section{Concluding remarks}\label{Concluding remarks}

We proposed a new family of inequality indices that bridges the Hoover index and the Gini coefficient through a continuous interpolation parameter. The index provides a unified framework that captures both deviations from the mean and pairwise differences, offering greater flexibility in inequality measurement. We established its main theoretical properties and derived analytical representations that enable explicit evaluation under gamma distributions. From a statistical perspective, we analyzed the plug-in estimator and obtained expressions for its expectation and bias. The Monte Carlo results show that the estimator performs well in finite samples, with improved accuracy as the sample size increases. The empirical application illustrates the practical relevance of the proposed index, highlighting its ability to provide a nuanced assessment of inequality and to serve as a complement to classical measures.

%\clearpage

\paragraph*{Acknowledgements}
The research was supported in part by CNPq and CAPES grants from the Brazilian government.

\paragraph*{Disclosure statement}
There are no conflicts of interest to disclose.

%%%%%%%%%%%%%%%%%%%%%%%%%%%%%%%%%%%%%%%%%%%%%%%%%%%%%%%%%%%%%

%\bibliographystyle{unsrt}
\bibliographystyle{apalike}
%	\bibliography{BibTexFileName}

%\appendix
%\section{Appendix}
%\label{sec:appendix_a}

\end{document}